\documentclass[conference,a4paper]{IEEEtran}

\usepackage[innermargin=0.545in,outermargin=0.545in,top=0.545in,bottom=.6in]{geometry}
\usepackage[utf8]{inputenc}
\usepackage[T1]{fontenc}
\usepackage{url}
\usepackage{ifthen}
\usepackage{cite}
\usepackage[cmex10]{amsmath} 
                             
\usepackage{tikz}
\usetikzlibrary{arrows,arrows.meta,matrix,positioning,positioning,calc}
\usepackage{pgfplots}
\pgfplotsset{compat=newest}

\usepackage{flushend}
\flushend
\usepackage{amsmath}
\usepackage{amssymb}
\usepackage{amsthm}
\usepackage{graphicx}

\newtheorem{theorem}{Theorem}
\newtheorem{lemma}{Lemma}

\newtheorem{remark}{Remark}

\newtheorem{definition}{Definition}
\newtheorem{construction}{Construction}

\def\ve#1{{\mathchoice{\mbox{\boldmath$\displaystyle #1$}}%
              {\mbox{\boldmath$\textstyle #1$}}%
              {\mbox{\boldmath$\scriptstyle #1$}}%
              {\mbox{\boldmath$\scriptscriptstyle #1$}}}}

\usepackage{xcolor}

\newcommand{\F}{\ensuremath{\mathbb{F}}}
\newcommand{\code}{\ensuremath{\mathcal{C}}}
\newcommand{\Code}{\ensuremath{\mathcal{C}}}

\renewcommand{\L}{\ensuremath{\mathcal{L}}}
\renewcommand{\O}{\ensuremath{\mathcal{O}}}
\newcommand{\R}{\ensuremath{\mathcal{R}}}
\newcommand{\W}{\ensuremath{\mathcal{W}}}

\newcommand{\s}{\ensuremath{\mathbf{s}}}
\newcommand{\sbar}{\ensuremath{\bar{\mathbf{s}}}}

\newcommand{\x}{\ensuremath{\mathbf{x}}}
\newcommand{\y}{\ensuremath{\mathbf{y}}}
\renewcommand{\H}{\ensuremath{\mathbf{H}}}
\newcommand{\C}{\ensuremath{\mathbf{C}}}
\newcommand{\0}{\ensuremath{\mathbf{0}}}
\renewcommand{\c}{\ensuremath{\mathbf{c}}}
\newcommand{\alphaVec}{\ensuremath{\ve{\alpha}}}
\newcommand{\gammaVec}{\ensuremath{\ve{\gamma}}}

\newcommand{\pmds}{\ensuremath{\mathsf{PMDS}}}
\newcommand{\sd}{\ensuremath{\mathsf{SD}}}
\newcommand{\Ha}{\ensuremath{\H^{(a)}}}
\newcommand{\order}{\mathcal{O}}

\newenvironment{bsmallmatrix}
  {\left[\begin{smallmatrix}}
  {\end{smallmatrix}\right]}

\title{Partial MDS Codes with Local Regeneration}

\author{%
  \IEEEauthorblockN{Lukas Holzbaur\IEEEauthorrefmark{1},
                    Sven Puchinger\IEEEauthorrefmark{2},
                    Eitan Yaakobi\IEEEauthorrefmark{3},
                    and Antonia Wachter-Zeh\IEEEauthorrefmark{1}}
  \IEEEauthorblockA{\IEEEauthorrefmark{1}%
                    Technical University of Munich,
                    \{lukas.holzbaur, antonia.wachter-zeh\}@tum.de}
  \IEEEauthorblockA{\IEEEauthorrefmark{2}%
                    Technical University of Denmark,
                    svepu@dtu.dk}
  \IEEEauthorblockA{\IEEEauthorrefmark{3}%
                    Technion --- Israel Institute of Technology,
                    yaakobi@cs.technion.ac.il}
  \thanks{The work of L. Holzbaur was supported by the Technical University of Munich -- Institute for Advanced Study, funded by the German Excellence Initiative and European Union 7th Framework Programme under Grant Agreement No. 291763 and the German Research Foundation (Deutsche Forschungsgemeinschaft, DFG) under Grant No. WA3907/1-1. Sven Puchinger has received funding from the European Union’s Horizon 2020 research and innovation programme under the Marie Skłodowska-Curie grant agreement no.~713683 (COFUNDfellowsDTU).}
}

\IEEEoverridecommandlockouts
\begin{document}

\maketitle

\begin{abstract}
  Partial MDS (PMDS) and sector-disk (SD) codes are classes of erasure codes that combine locality with strong erasure correction capabilities. We construct PMDS and SD codes where each local code is a bandwidth-optimal regenerating MDS code. The constructions require significantly smaller field size than the only other construction known in literature.
\end{abstract}

\section{Introduction}

Distributed data storage is ever increasing its importance with the amount of data stored by cloud service providers and data centers in general reaching staggering heights. The data is commonly spread over a number of nodes (servers or hard drives) in a \emph{distributed storage system} (DSS), with some additional redundancy to protect the system from data loss in the case of node failures (erasures). The resilience of a DSS against such events can be measured either by the minimal \emph{number of nodes} that needs to fail for data loss to occur, i.e., the \emph{distance} of the storage code, or by the expected time the system can be operated before a failure occurs that causes data loss, referred to as the \emph{mean time to data loss}. For both measures the use of maximum distance separable (MDS) codes provides the optimal trade-off between storage overhead and resilience to data loss (note that replication is a trivial MDS code). The downside of using MDS codes is the cost of recovering (replacing) a failed node. Consider a storage system with $k$ information nodes and $s$ nodes for redundancy. If an MDS code is used for the recovery of a node by means of erasure decoding, it  necessarily involves at least $k$ nodes (helpers) and, if done by straight-forward methods, a large amount of network traffic, namely the download of the entire content from $k$ nodes. To address these issues, the concepts of \emph{locally repairable codes} (LRCs) \cite{gopalan2012locality,kamath2014codes,rawat2013TiT,Krishnan2018,gligoroski2017locally,hollmann2014minimum,li2016relieving} and \emph{regenerating codes} \cite{dimakis2010network,cadambe2013asymptotic,ye2017optimalRepair} have been introduced.

To lower the amount of network traffic in recovery, regenerating codes allow for repairing nodes by accessing $d > k$ nodes, but only retrieve a function of the data stored on each node. This significantly decreases the repair traffic. Lower bounds on the required traffic for repair have been derived in \cite{dimakis2010network,cadambe2013asymptotic} which lead to two extremal code classes, namely \emph{minimum bandwidth regenerating} (MBR) and \emph{minimum storage regenerating} (MSR) \emph{codes}. MBR codes offer the lowest possible repair traffic, but at the cost of increased storage overhead compared to MDS codes. In this work we consider $d$-MSR codes, which require more network traffic for repair than MBR codes, but are optimal in terms of storage overhead, i.e., they are MDS.

To address the other downside of node recovery in MDS codes, namely the large number of required helper nodes, LRCs introduce additional redundancy to the system, such that in the (more likely) case of a few node failures the recovery only involves less than $k$ helper nodes, i.e., can be performed \emph{locally}. This subset of helper nodes is referred to as a \emph{local code}. Recently several constructions of LRCs which maximize the distance have been proposed. However, when considering the mean time to data loss as the performance metric, distance-optimal LRCs are not necessarily optimal, as it is possible to tolerate many failure patterns involving a larger number of nodes than the number that can be guaranteed, while still fulfilling the locality constraints \cite{tamo2016optimal,holzbaur2019error}. \emph{Partial MDS} (PMDS) \emph{codes} \cite{blaum2013partial,blaum2016construction,gabrys2018constructions}, also referred to as \emph{maximally recoverable codes} \cite{chen2007maximally,gopalan2014explicit}, are a subclass of LRCs which guarantee to tolerate \emph{all} failure patterns possible under these constraints and thereby maximize the mean time to data loss. Specifically, an $(r,s)$-PMDS code of length $\mu n$ can be partitioned into $\mu$ local groups of size $n$, such that any erasure pattern with $r$ erasures in each local group plus any $s$ erasures in arbitrary positions can be recovered.

However, the local recovery of nodes still induces a large amount of network traffic, as the entire content of the helper nodes needs to be downloaded when considering straight-forward use recovery algorithms. To circumvent this bottleneck, several regenerating local codes \cite{dimakis2010network} have been proposed \cite{kamath2014codes,rawat2013TiT,Krishnan2018,gligoroski2017locally,hollmann2014minimum,li2016relieving}. In \cite{calis2016} it was shown that the LRC construction of \cite{rawat2013TiT} is in fact a PMDS code, implicitly giving the first construction of PMDS codes with local regeneration\footnote{The construction in \cite{rawat2013TiT} consists of two encoding stages, where in the second stage an arbitrary linear MDS code can be used to obtain the local codes. In \cite{calis2016} it was shown that the construction in fact gives a PMDS code, independent of the explicit choice of the MDS code in the second encoding stage. It follows that using a regenerating MDS code in the second encoding stage results in a PMDS code with local regeneration.}. However, these PMDS codes require a field size exponential in the length of the code. 

In this work, we propose the first constructions of locally regenerating PMDS codes with field size that is not exponential in the length. Our first construction gives $(r,s=2)$-PMDS codes where each local code is a $d$-MSR code by a non-trivial combination of the $(r,s=2)$-PMDS codes of \cite{blaum2016construction} and the MSR codes of \cite{ye2017optimalRepair}, over a field of size $O(rn)$. The second construction of $(r,s)$-PMDS codes with local $d$-MSR codes, which is valid for any value of $s$, combines the $(r,s)$-PMDS construction of \cite{gabrys2018constructions} with the MSR codes of \cite{ye2017optimalRepair} and requires a field size of $O((\mu n)^{s(r+1)})$, when $r=O(1)$ and $s = O(1)$. The number of field elements stored at each node (\emph{subpacketization}) is $\ell=O(r^n)$, and thus equal to the subpacketization in the  MSR code construction of \cite{ye2017optimalRepair}.

\section{Preliminaries}

\subsection{Notation}
We write $[a,b]$ for the set of integers $\{a,a+1,...,b\}$ and $[b]$ if $a=1$.
For a set of integers $R \subseteq [n]$ and a code $\code$ of length $n$ we write $\code |_R$ for the code obtained by restricting $\code$ to the positions indexed by $R$, i.e., puncturing in the positions $[n]\setminus R$. For an element $\alpha \in \F$ we denote its order by $\O(\alpha)$.

We denote a code of length $n$, dimension $k$, and distance $d_{\min}$ by $[n,k,d_{\min}]$. For a code over $\F_{q^\ell}$ that is linear over $\F_q$ we write $[n,k,d_{\min};\ell]$. The parameter $\ell$ is referred to as the subpacketization of the code and as each codeword of this code can be viewed as an array over $\F_q$ with $n$ columns and $\ell$ rows, we also refer to such codes as \emph{array codes}. If the distance $d_{\min}$ is clear from context or not of interest, we omit it from the notation.

This work is largely based on the constructions of PMDS codes in \cite{blaum2016construction,gabrys2018constructions} and the construction of MSR codes in \cite{ye2017optimalRepair}. Since the notations in these works are conflicting, i.e., the same symbols are used for different parameters of the codes, Table~\ref{tab:notation} provides an overview of the notation used in this work compared to these works.
\begin{table}[htb]
  \centering
  \caption{Overview of the notation used in this work compared to the notation used in \cite{blaum2016construction,gabrys2018constructions,ye2017optimalRepair}. As we are constructing PMDS codes with local MSR codes in the following, the length and number of parities in the MSR code construction of \cite{ye2017optimalRepair} are matched with the parameters of the local codes in our work. }
  \begin{tabular}{lcccc}
    Description &\cite{blaum2016construction}&\cite{gabrys2018constructions}& \cite{ye2017optimalRepair} & This work \\ \hline
    Number of local groups & $r$ & $m$ & - & $\mu$ \\
    Length of local MSR code & $n$ & $n$ & $n$ & $n$ \\
    Number of local parity symbols & $m$ & $r$ & $r$ & $r$ \\
    Number of global parity symbols & $s$ & $s$ & - & $s$ \\
    Code length & $rn$ & $mn$ & -  & $\mu n$ \\
    Subpacketization & - & - & $l$ & $\ell$ \\
    Number of nodes needed for repair & - & - & $d$ & $d$
  \end{tabular}
  \label{tab:notation}
\end{table}
\subsection{Definitions}

We begin by formally defining PMDS codes in our notation.
\begin{definition}[Partial MDS array codes]\label{def:pmds}
  Let $\mu n,\mu,n,r,s,\ell \in \mathbb{Z}_{>0}$ such that $r\leq n$ and $s\leq (n-r)\mu$.

  Let $\code \subset \F_q^{\ell \times \mu n}$ be a linear $[\mu n,(n-r)\mu -s;\ell]$ code. The code $\code$ is a $\pmds(\mu n,\mu,n,r,s;\ell)$ \emph{partial MDS array code} if there exists a partition $\mathcal{W} = \{W_1,W_2,...,W_{\mu}\}$ of $[\mu n]$ with $|W_i|=n$ for all $i\in [\mu]$ such that
  \begin{itemize}
  \item the code $\code |_{W_i}$ is an $[n,n-r,r+1; \ell]$ MDS code for all $i \in [\mu]$ and
  \item the code $\code |_{[\mu n] \setminus \cup_{i=1}^{\mu} E_i}$ is an $[n-r\mu,n-r\mu-s,s+1;\ell]$ MDS code, where $E_i \subset W_i$ with $|E_i|=r$ for all $i\in [\mu]$.
  \end{itemize}
\end{definition}
We refer to the code $\code |_{W_i}$ as the $i$-th \emph{local code}.
\begin{remark}
In \cite{blaum2016construction,gabrys2018constructions} each codeword of the PMDS and SD codes is regarded as an $\mu \times n$ array, where for PMDS codes each row can correct $r$ erasures and for SD codes $r$ erased columns can be corrected. As we will construct PMDS and SD codes with local MSR codes, we will require subpacketization, i.e., each node will not store a symbol, but a vector of multiple symbols. To avoid having different types of rows, we adopt the terminology commonly used in the LRC literature and view the codewords of a PMDS or SD code as vectors, and what we refer to as \emph{local codes} is equivalent to the \emph{rows} of \cite{blaum2016construction,gabrys2018constructions}.
\end{remark}

In the following we will construct both PMDS and SD codes with local regeneration, but since the concepts and proofs are mostly the same, we provide them in less detail for SD codes.
\begin{remark}\label{rem:SDcodes}
  Sector-Disk codes are defined similar to PMDS codes as in Definition~\ref{def:pmds}, except that $E_1=E_2= \ldots =E_{\mu}$ holds.
\end{remark}

\begin{definition}[Regenerating code \cite{dimakis2010network,cadambe2013asymptotic}]
  Let $\mathcal{F}, \R \subset [n]$ be two disjoint subsets with $\mathcal{F} \leq r$ and $\R \geq n-r$. Let $\code$ be an $[n,n-r;\ell]$ MDS array code $\code$ over $\F_q$. Define $M(\code,\mathcal{F},\R)$ as the smallest number of symbols of $\F$ one needs to download from the surviving nodes indexed by $\R$ to recover the erased nodes indexed by $\mathcal{F}$. Then
  \begin{equation}\label{eq:boundRegenrating}
    M(\code,\mathcal{F},\R) \geq \frac{|\mathcal{F}||\R|\ell}{|\mathcal{F}|+|\R|-n+r} \ .
  \end{equation}
  We say that the code $\code$ is an \emph{$(h,d)$-MSR} code if
  \begin{equation*}
    \max_{\substack{|\mathcal{F}| = h, |\R|=d \\ \mathcal{F} \cap \R = \emptyset}} M(\code,\mathcal{F},\R) = \frac{|\mathcal{F}||\R|\ell}{|\mathcal{F}|+|\R|-n+r} \ .
  \end{equation*}
  If $h=1$ we say the code is a $d$-MSR code.
\end{definition}

\begin{definition}[Locally $(h,d)$-MSR PMDS array codes]\label{def:locallyMSR}
  Let $\code$ be a $\pmds(\mu n,\mu,n,r,s;\ell)$ code with partition $\W$. We say that the code $\code$ is locally $(h,d)$-MSR if $\code |_{W_i}$ is an $(h,d)$-MSR code for all $i\in [\mu]$.
  If $h=1$ we say the code is a locally $d$-MSR PMDS code.
\end{definition}

\begin{figure}
  \centering
    \resizebox{\columnwidth}{!}{\def\x{0.5}

\begin{tikzpicture}

\node (S1) at (0,0) [draw,thick,minimum width=\x*0.75cm,minimum height=\x*6.5cm] {};
\node (S2)  [right=\x*0.3cm of S1, draw,thick,minimum width=\x*0.75cm,minimum height=\x*6.5cm] {};
\node (S3)  [right=\x*0.3cm of S2, draw,thick,minimum width=\x*0.75cm,minimum height=\x*6.5cm] {};
\node (S32)  [right=\x*0.3cm of S3, draw,thick,minimum width=\x*0.75cm,minimum height=\x*6.5cm] {};
\node (S4)  [right=\x*0.3cm of S32, draw,thick,minimum width=\x*0.75cm,minimum height=\x*6.5cm] {};


\node (S5)  [right=\x*0.7cm of S4, draw,thick,minimum width=\x*0.75cm,minimum height=\x*6.5cm] {};
\node (S6)  [right=\x*0.3cm of S5, draw,thick,minimum width=\x*0.75cm,minimum height=\x*6.5cm] {};
\node (S7)  [right=\x*0.3cm of S6, draw,thick,minimum width=\x*0.75cm,minimum height=\x*6.5cm] {};
\node (S72)  [right=\x*0.3cm of S7, draw,thick,minimum width=\x*0.75cm,minimum height=\x*6.5cm] {};
\node (S8)  [right=\x*0.3cm of S72, draw,thick,minimum width=\x*0.75cm,minimum height=\x*6.5cm] {};


\node (S9)  [right=\x*0.7cm of S8, draw,thick,minimum width=\x*0.75cm,minimum height=\x*6.5cm] {};
\node (S10)  [right=\x*0.3cm of S9, draw,thick,minimum width=\x*0.75cm,minimum height=\x*6.5cm] {};
\node (S11)  [right=\x*0.3cm of S10, draw,thick,minimum width=\x*0.75cm,minimum height=\x*6.5cm] {};
\node (S112)  [right=\x*0.3cm of S11, draw,thick,minimum width=\x*0.75cm,minimum height=\x*6.5cm] {};
\node (S12)  [right=\x*0.3cm of S112, draw,thick,minimum width=\x*0.75cm,minimum height=\x*6.5cm] {};


\draw[dotted,thick] (\x*5.2, \x*3) -- (\x*5.2,-\x*3);
\draw[dotted,thick] (\x*11.1, \x*3) -- (\x*11.1,-\x*3);


\foreach \i in {1,...,12,32,72,112}{
  \foreach \j in {0,1,2,3}{
  \draw ($(S\i)+(-\x*0.25,\x*3.15-\x*\j*0.4)$) rectangle ($(S\i)+(+\x*0.25,\x*2.85-\x*\j*0.4)$) node (C\i) {};}
  \foreach \j in {0,1,2,3}{
  \draw ($(S\i)+(-\x*0.25,\x*1.35-\x*\j*0.4)$) rectangle ($(S\i)+(+\x*0.25,\x*1.05-\x*\j*0.4)$) node (C\i) {};}
  \foreach \j in {0,1,2,3}{
  \draw ($(S\i)+(-\x*0.25,-1.65*\x-\j*0.2)$) rectangle ($(S\i)+(+\x*0.25,-1.95*\x-\j*0.2)$) node (C\i) {};}
  \node (C\i) at ($(S\i)+(0,\x*-0.7)$) [minimum width=\x*0.5cm,minimum height=\x*0.2cm,rounded corners=1pt] {$\vdots$};

}

\draw[dashed, rounded corners = 1pt, blue, thick] ($(S1)+(-\x*.5,\x*3.3)$) rectangle ($(S12)+(\x*.5,\x*2.7)$);


\node[anchor = south east] (L1) at ($(S12)+(\x*1,\x*4.5)$) {\footnotesize PMDS codeword};
\path ($(L1)+(\x*2.1,-\x*0.5)$) edge[bend left, -{Latex[length=1mm,width=0.8mm]}]  ($(S12)+(\x*.7,\x*3.3)$) ;

\draw[dashed, rounded corners = 1pt, orange, thick] ($(S1)+(-\x*.5,\x*3.3)$) rectangle ($(S4)+(\x*.5,\x*1.5)$);
\node[anchor = south east] (L2) at ($(S4)+(\x*1,\x*4.4)$) {\footnotesize Ye-Barg codeword};
\path ($(L2)+(\x*2.1,-\x*0.5)$) edge[bend left, -{Latex[length=1mm,width=0.8mm]}]  ($(S4)+(\x*0.55,\x*2.3)$) ;

\draw [decorate,decoration={brace,amplitude=2pt}]  ($(S1)+(-\x*0.7,\x*1.5)$) -- ($(S1)+(-\x*0.7,\x*3.3)$) node [black,midway,xshift=-0.25cm] {\footnotesize $\ell$ };

\draw[thick,-{Latex[length=2mm,width=1mm]},red] ($(S2)+(0,\x*3)$) -- ($(S2)+(-\x*0.2,-\x*0.4)$) -- ($(S2)+ (\x*0.2,\x*0.4)$) -- ($(S2)+(0,-\x*3)$);
\draw[thick,-{Latex[length=2mm,width=1mm]},red] ($(S7)+(0,\x*3)$) -- ($(S7)+(-\x*0.2,-\x*0.4)$) -- ($(S7)+ (\x*0.2,\x*0.4)$) -- ($(S7)+(0,-\x*3)$);
\draw[thick,-{Latex[length=2mm,width=1mm]},red] ($(S9)+(0,\x*3)$) -- ($(S9)+(-\x*0.2,-\x*0.4)$) -- ($(S9)+ (\x*0.2,\x*0.4)$) -- ($(S9)+(0,-\x*3)$);
\draw[thick,-{Latex[length=2mm,width=1mm]},red] ($(S10)+(0,\x*3)$) -- ($(S10)+(-\x*0.2,-\x*0.4)$) -- ($(S10)+ (\x*0.2,\x*0.4)$) -- ($(S10)+(0,-\x*3)$);
\draw[thick,-{Latex[length=2mm,width=1mm]},red] ($(S11)+(0,\x*3)$) -- ($(S11)+(-\x*0.2,-\x*0.4)$) -- ($(S11)+ (\x*0.2,\x*0.4)$) -- ($(S11)+(0,-\x*3)$);
\draw[thick,-{Latex[length=2mm,width=1mm]},red] ($(S4)+(0,\x*3)$) -- ($(S4)+(-\x*0.2,-\x*0.4)$) -- ($(S4)+ (\x*0.2,\x*0.4)$) -- ($(S4)+(0,-\x*3)$);
\draw[thick,-{Latex[length=2mm,width=1mm]},red] ($(S112)+(0,\x*3)$) -- ($(S112)+(-\x*0.2,-\x*0.4)$) -- ($(S112)+ (\x*0.2,\x*0.4)$) -- ($(S112)+(0,-\x*3)$);


\node at ($(S6)+(\x*1.1,\x*4.7)$) {Servers};
\node at ($(S2)+(\x*1.1,\x*3.8)$) {\footnotesize Local group $1$};
\node at ($(S6)+(\x*1.1,\x*3.8)$) {\footnotesize Local group $2$};
\node at ($(S10)+(\x*1.1,\x*3.8)$) {\footnotesize Local group $3$};


\node (El) at ($(S6)+(\x*0.6,-\x*4.3)$) {\color{red} Server Failures};
\path ($(El.north west)+(\x*0,\x*-0.6)$) edge[bend left=5, -{Latex[length=1mm,width=0.8mm]}]  ($(S2)+(\x*0.1,-\x*3.4)$) ;
\path ($(El.north west)+(\x*0,\x*-0.2)$) edge[bend left=5, -{Latex[length=1mm,width=0.8mm]}]  ($(S4)+(\x*0.1,-\x*3.4)$) ;
\path ($(El.north) + (\x*0.5,0)$) edge[bend right=5, -{Latex[length=1mm,width=0.8mm]}]  ($(S7)+(-\x*0,-\x*3.4)$) ;
\path ($(El.north east)+ (\x*0,-\x*0.3)$) edge[bend right=5, -{Latex[length=1mm,width=0.8mm]}]  ($(S9)+(\x*0.1,-\x*3.4)$) ;
\path ($(El.north east)+ (\x*0,-\x*0.4)$) edge[bend right=5, -{Latex[length=1mm,width=0.8mm]}]  ($(S10)+(\x*0.1,-\x*3.4)$) ;
\path ($(El.north east)+ (\x*0,-\x*0.5)$) edge[bend right=5, -{Latex[length=1mm,width=0.8mm]}]  ($(S11)+(\x*0.1,-\x*3.4)$) ;
\path ($(El.north east)+(\x*0,\x*-0.6)$) edge[bend right=5, -{Latex[length=1mm,width=0.8mm]}]  ($(S112)+(-\x*0.1,-\x*3.4)$) ;

\end{tikzpicture}

}
  \caption{Illustration of locally regenerating PMDS array codes as constructed in this work, with $n=5$, $\mu=3$ and each symbol of the code alphabet represented by a small rectangle. The shown erasure pattern can be corrected by an $(r=2,s=2)$-PMDS code. }
  \label{fig:illustration}
\end{figure}
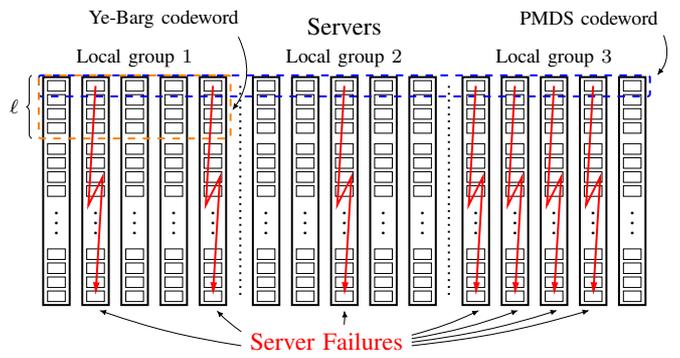
Figure \ref{fig:illustration} shows an illustration of a locally regenerating PMDS array code. Assuming it to be an $(r=2,s=2)$-PMDS code, the erasures in the first local code can be corrected locally, but without taking advantage of the regenerating property, as the number of available helper nodes is only $n-r$. The erasure in the second local code can be corrected from the remaining $n-r+1$ nodes in the local group using the locally regenerating property, and the erasures in the third local code can be recovered by accessing nodes of the other local groups. Note that the example was chosen specifically to illustrate these different cases, the case of a single erasure in a local code, for which the locally regenerating property decreases the repair bandwidth, is far more likely than the other cases.

\subsection{Ye-Barg Regenerating Codes}

We repeat \cite[Construction~2]{ye2017optimalRepair} in the slightly different notation which will be used in this work. 

\begin{definition}[Ye-Barg $d$-MSR codes {\cite[Construction~2]{ye2017optimalRepair}}] \label{def:yeBarg}
  Let $\code \subset \F_q$ be an $[n,n-r;\ell]$ array code over $\F_q$, where $q\geq b n$ and $b=d+1-n+r$. Let $\{\beta_{i,j} \}_{i\in [n], j\in [b]}$ be a set of $bn$ distinct elements of $\F_q$. Then each codeword is an array with $\ell = b^{n}$ rows and $n$ columns, where the $i$-th row fulfills the parity check equations
  \begin{equation*}
    \Ha =
      \begin{bsmallmatrix}
        1&1& \hdots & 1 \\
        \beta_{1,a_1} & \beta_{2,a_2} & \hdots & \beta_{n,a_{n}}\\
        \vphantom{\int\limits^x}\smash{\vdots} & \vphantom{\int\limits^x}\smash{\vdots} & & \vphantom{1}\smash{\vdots} \\
        \beta_{1,a_1}^{r-1} & \beta_{2,a_2}^{r-1} & \hdots & \beta_{n,a_{n}}^{r-1}
      \end{bsmallmatrix}\ ,
  \end{equation*}
  for $a \in [0,\ell-1]$ and $a = \sum_{i = 1}^{n} a_i b^{i-1}$.
\end{definition}
\begin{remark}
  The constructions presented in Section~\ref{sec:constructionS2} and~\ref{sec:constructionAny} can also be applied to obtain locally $(h,d)$-MSR PMDS codes, where each local code is an $(h,d)$-MSR code as in \cite[Construction~3]{ye2017optimalRepair}, which is very similar in structure to \cite[Construction~2]{ye2017optimalRepair} given in Definition~\ref{def:yeBarg}. However, as the required subpacketization is larger for the former, we focus on $d$-MSR codes in this work. 
\end{remark}
\begin{remark}
  In Definition~\ref{def:yeBarg} we define each row of the array code by a set of parity check equations independent of the other $\ell-1$ rows of the array. Note that this is not possible for array codes in general. However, for the existence of such a description it is sufficient that the matrices $A_i$, as defined in \cite{ye2017optimalRepair}, are diagonal matrices. This simplifies the notation for the cases considered in this work, as this notation makes it obvious that each row is an $[n,n-r]$ RS code, and thereby MDS.
\end{remark}

\section{Regenerating PMDS and Sector-Disk codes with Two Global Parities} \label{sec:constructionS2}

We construct array codes from the PMDS and SD codes of \cite{blaum2016construction} using the ideas of \cite{ye2017optimalRepair} to obtain locally $d$-MSR PMDS codes.

\subsection{Generalization of known PMDS construction} \label{sec:BlaumSDandPMDS}

 To apply the ideas of \cite{ye2017optimalRepair} when constructing locally $d$-MSR PMDS and SD codes we need the local codes to be RS codes with specific code locators. The construction of PMDS codes given in \cite{blaum2016construction} has the property that the local codes are RS codes, but the code locators are fixed to be the first $n$ powers of some element $\beta$ of sufficient order. We generalize this construction to allow for different choices of code locators for the local codes.

  Let $\beta \in \F_{2^w}$ be an element of order $\order(\beta) \geq \mu N$. The $[\mu n,\mu(n-r)-2]$ code $\code(\mu ,n,r,2,\mathcal{L},N)$ is given by the $(r\mu +2)\times \mu n$ parity-check matrix 
  \begin{align*}
   \H =
   \begin{bsmallmatrix}
      \H_0 & \mathbf{0} &\hdots & \mathbf{0} \\
      \mathbf{0} & \H_0 &\hdots & \mathbf{0} \\
      \vdots & \vdots &\ddots & \cdots \\
      \mathbf{0}& \mathbf{0} &\hdots & \H_0 \\
      \H_1 & \H_2 &\hdots & \H_{\mu } \\
    \end{bsmallmatrix} 
  \end{align*}
  where
  \begin{align*}
    \H_0  =
 \begin{bsmallmatrix}
     1 & 1 & \hdots & 1 \\
     \beta^{i_1} & \beta^{i_2} & \hdots & \beta^{i_{n}} \\
     \beta^{2i_1} & \beta^{2i_2} & \hdots & \beta^{2i_{n}}\\
     \vdots & \vdots & \ddots & \vdots \\
     \beta^{(r-1)i_1} & \beta^{(r-1)i_2} & \hdots & \beta^{(r-1)i_{n}}
    \end{bsmallmatrix}
  \end{align*}
  for $\L = \{i_1,i_2,...,i_{n}\}$ and, for $0\leq j\leq \mu-1$,
  \begin{align*}
    \H_{j+1} =
     \begin{bsmallmatrix}
      \beta^{ri_1} & \beta^{ri_2} &  \hdots & \beta^{rn} \\
      \beta^{-jN-i_1} & \beta^{-jN-i_2} & \hdots & \beta^{-jN-i_{n}}
    \end{bsmallmatrix} \ .
  \end{align*}

Note that this generalization includes both \cite[Construction~A]{blaum2016construction} and \cite[Construction~B]{blaum2016construction} as special cases
\begin{equation*}
\code_A(\mu ,n,r,2,\{0,1,..., n-1\},n)
\end{equation*}
and
\begin{equation*}
\code_B(\mu ,n,r,2,\{0,1,..., n-1\},N_B)
\end{equation*}
for $N_B = (r+1)(n-1-r)+1$, respectively.

We now derive a general, sufficient condition on $N$, based on the set $\L$, such that the code is a PMDS code.

\begin{lemma} \label{lem:PMDSgeneral}
  Let $\L$ be any set of non-negative integers with $|\L| = n$, then the code $\code(\mu ,n,r,2,\L,N)$ is a PMDS code for any $N \geq (r+1) (\max_{i \in \L}i-r)+1$
\end{lemma}
\begin{proof}
We follow the proof of \cite[Theorem~5]{blaum2016construction}. Assume $r$ positions in each local group (row of the PMDS code) have been erased and in addition there are $2$ random erasures. If the two erasures occur in the same local group $z$, all local groups except for this one will be corrected by the local codes. Since all points in $\L$ are distinct, by the same argument as in \cite{blaum2016construction}, the matrix corresponding to the erased positions is a Vandermonde matrix and the erasures can be corrected.\\
Now consider the case of two local groups (horizontal codes) with $r+1$ erasures each, where the erased positions are given by $\{j_1,...,j_{r+1}\} \subset \L$ and $\{j_1',...,j_{r+1}'\} \subset \L$, respectively. By the same arguments as in \cite[Theorem~7]{blaum2016construction} we obtain, that such an erasure pattern is correctable if
  \begin{equation*}
    Nz + \sum_{u=1}^{r+1} j_u' - \sum_{u=1}^{r+1} j_u' \neq 0 \mod \order(\beta) \ .
  \end{equation*}
  Since all $j_u$ are distinct, we have
  \begin{equation*}
    \frac{r(r+1)}{2} = \sum_{u=0}^{r} u \leq \sum_{u=1}^{r+1} j_u
  \end{equation*}
  and
  \begin{align*}
    \sum_{u=1}^{r+1}j_u &\leq \sum_{u=0}^{r}(\max_{ j\in \L}j -r)+u \\
                    &= (r+1)(\max_{j\in \L}j-r) + \sum_{u=0}^{r} u = N-1 + \frac{r(r+1)}{2} \ .
  \end{align*}
  The remaining steps are exactly the same as in \cite[Theorem~7]{blaum2016construction} and we obtain
  \begin{align*}
   1 \leq Nz+ \sum_{u=1}^{r+1} j_u' - \sum_{u=1}^{r+1} j_u \leq N\mu -1 < \order(\beta)
  \end{align*}
  and the lemma statement follows.
\end{proof}

By similar arguments we also give a general, sufficient condition on $N$ for the code to be an SD code.
\begin{lemma} \label{lem:SDgeneral}
  Let $\L$ be any set of non-negative integers with $|\L| = n$, then the code $\code(\mu ,n,r,2,\L,N)$ is an SD code for any $N \geq \max_{j \in \L} j+1$
\end{lemma}
\begin{proof}
  The case of $r+2$ erasures in the same local group (horizontal code) is the same as in Lemma~\ref{lem:PMDSgeneral} and \cite[Theorem~5/7]{blaum2016construction}.  Now consider the case of $r$ column erasures in positions $j_1,...,j_{r} \in \L$ and two random erasures in distinct rows $z$ and $z'$ in positions $j,j' \in \L \setminus \{j_1,...,j_{r}\}$. We assume without loss of generality that $z < z'$. By the same argument as in \cite{blaum2016construction} we need to show that $\beta^{-j} + \beta^{-N(z-z')-j'}$ is invertible. With $1\leq z,z' \leq \mu$ and $0\leq j,j' \leq N-1$ we get
  \begin{align*}
    N(z'-z)+j'-j \geq N +j'-j \geq N - (N-1) > 0
  \end{align*}
  and
  \begin{align*}
    N(z'-z)+j'-j \leq N(\mu -1) + N-1 = N\mu  -1 < \order(\beta) \ .
  \end{align*}
  Combining these we get $1\leq N(z'-z)+j'-j \leq N\mu -1$, so
  \begin{align*}
    N(z'-z)+j'-j \neq 0 \mod \order(\beta)
  \end{align*}
  and it follows that $\beta^{-j} \oplus \beta^{-N(z-z')-j'}$ is invertible.
\end{proof}

With these generalizations of \cite[Construction~A/B]{blaum2016construction} we are now ready to construct PMDS and SD codes, where each local code is a $d$-MSR code.
\begin{construction}[Locally $d$-MSR PMDS/SD array codes]\label{con:SDcodesLocalRegeneration}
Let $s=2$ and $q,\mu,n,r,d,N \in \mathbb{Z}_{>0}$ be positive integers with
\begin{itemize}
    \item $r \leq n$
    \item $q$ a power of $2$
    \item $q \geq \max\{\mu N, bn\}$, where $b=d+1-(n-r)$
    \item $\ell = b^n$
\end{itemize}
 For an element $\beta \in \F_q$ with $\order(\beta) \geq \max\{\mu N,nb\}$ denote $\beta_{i,j} = \beta^{i-1+(j-1)n}, 1\leq i\leq n, 1\leq j\leq b$.

We define the following $[\mu n, \mu(n-r)-2; \ell]_{q^M}$ array code $\code(\mu ,n,r,2,N,d;\ell)_{q}$ as
\begin{align*}
\left\{
\C = \begin{bsmallmatrix}
\c^{(0)} \\
\c^{(1)}\\
\vdots 	   \\
\c^{(\ell-1)} 
\end{bsmallmatrix} \in \F_{q}^{\ell \times \mu n} \, : \, \H^{(a)} \c^{(a)} = \0 \, \forall \, a=0,\dots,\ell-1
\right\},
\end{align*}
The matrix $\H^{(a)}$ is defined as
  \begin{align*}
    \H^{(a)} =
    \begin{bsmallmatrix}
      \H_0^{(a)} & \mathbf{0} &\hdots & \mathbf{0} \\
      \mathbf{0} & \H_0^{(a)} &\hdots & \mathbf{0} \\
      \vdots & \vdots &\ddots & \cdots \\
      \mathbf{0}& \mathbf{0} &\hdots & \H_0^{(a)} \\
      \H_1^{(a)} & \H_2^{(a)} &\hdots & \H_{\mu }^{(a)} \\
    \end{bsmallmatrix} \in \F_{q}^{r \mu +2 \times \mu n} \ ,
    \end{align*}
    where
    \begin{align}\label{eq:parityCheckLocalS2}
      \H_0^{(a)}  =
      \begin{bsmallmatrix}
        1&1& \hdots & 1 \\
        \beta_{1,a_1} & \beta_{2,a_2} & \hdots & \beta_{n,a_{n}}\\
        \vphantom{\int\limits^x}\smash{\vdots} & \vphantom{\int\limits^x}\smash{\vdots} & & \vphantom{1}\smash{\vdots} \\
        \beta_{1,a_1}^{r-1} & \beta_{2,a_2}^{r-1} & \hdots & \beta_{n,a_{n}}^{r-1} 
      \end{bsmallmatrix} \in \F_{q}^{r \times n} \ ,
  \end{align}
  with $a \in [0,\ell-1]$ and $a = \sum_{i = 1}^{n} a_{i} b^{i-1}$.
   For $0\leq j\leq \mu -1$ let
  \begin{align*}
    \H_{j+1}^{(a)} =
\begin{bsmallmatrix}
      \beta_{1,a_1}^r & \beta_{2,a_2}^r & \hdots & \beta_{n,a_{n}}^r \\
      \beta^{-jN} \beta_{1,a_1}^{-1} & \beta^{-jN} \beta_{2,a_2}^{-1}  & \hdots & \beta^{-jN} \beta_{n,a_{n}}^{-1}
    \end{bsmallmatrix} \in \F_{q}^{2 \times n} \ .
  \end{align*}
\end{construction}

It remains to show that the local codes are MSR codes and the conditions under which the code is a PMDS or SD code. 
\begin{theorem}\label{thm:PMDSMSR}
Let $\mu,n,r,d$ be fixed and $q \geq \max\{\mu N, bn\}$ with
  \begin{equation*}
    N=(r+1)(rn-1-r)+1 \ .
  \end{equation*}
  Then the code $\code(\mu ,n,r,2,N,d;\ell)_{q}$ as in Construction~\ref{con:SDcodesLocalRegeneration} is a locally $d$-MSR $\pmds(\mu n,\mu,n,r,2;b^{n})$ code over $\F_q$, as in Definition~\ref{def:locallyMSR}.
\end{theorem}
\begin{proof}
  First, note that the $\beta_{i,j}$ in Construction~\ref{con:SDcodesLocalRegeneration} are the (distinct) elements $\beta^0,\beta^1,...,\beta^{rn-1}$. Now consider the $j$-th local group. The $a$-th row fulfills the parity check equations given in (\ref{eq:parityCheckLocalS2}) and since all elements $\beta_{i,j}$ are distinct, it is immediate that the local group is an $[n,n-r; b^{n}]$ Ye-Barg code as in Definition~\ref{def:yeBarg}.

  For the PMDS property, observe that the $a$-th row, i.e., the row fulfilling the parity-check equations $\H^{(a)}$, is a code $\code(\mu,n,r,2,\L^{(a)},N)$ as in Section~\ref{sec:BlaumSDandPMDS}, where $\L^{(a)} = \{i-1+(a_i-1)n \ | \ i \in [n]\}$ by definition of the $\beta_{i,j}$. For any $a$ it holds that
  \begin{align*}
    \max_{i \in \L^{(a)}} i \leq \max_{\substack{i \in \L^{(a)}\\a\in [0,\ell-1]}} i = rn -1 \ .
  \end{align*}
  By Lemma~\ref{lem:PMDSgeneral} a code as in Section~\ref{sec:BlaumSDandPMDS} is PMDS if $N \geq (r+1)(\max_{i\in \L}i -r) +1$ and the lemma statement follows.
\end{proof}

\begin{theorem}\label{thm:SDMSR}
  Let $\mu,n,r,d$ be fixed and $q \geq \max\{rn\mu, bn\}$. Then the code $\code(\mu ,n,r,2,rn,d;\ell)_{q}$ as in Construction~\ref{con:SDcodesLocalRegeneration} is a locally $d$-MSR $\sd(\mu n,\mu,n,r,s;b^{n})$ code over $\F_q$.
\end{theorem}
\begin{proof}
  The proof follows immediately from the proof of Theorem~\ref{thm:PMDSMSR} and Lemma~\ref{lem:SDgeneral}.
\end{proof}

\section{PMDS codes with local regeneration and any number of global parities} \label{sec:constructionAny}

We now consider the more general problem of constructing locally $d$-MSR PMDS and SD codes for any number of local and global parities, based on the construction of PMDS and SD codes in \cite{gabrys2018constructions} and the construction of $d$-MSR codes in \cite{ye2017optimalRepair}. 

\subsection{Recapitulation and Generalization of the PMDS Construction in \cite[Section~III]{gabrys2018constructions}}\label{ssec:recap_PMDS_arbitrary_r_s}

We first rephrase the construction of the PMDS codes of \cite[Section~III]{gabrys2018constructions} in our notation and slightly generalize it.

Let $q$ be a prime power and $M$ be a positive integer. Furthermore, let $\H_0 \in \F_q^{r \times n}$ be a parity-check matrix of an $[n,n-r]_{q}$ MDS code and $\alpha_{i,j}$ be elements of the extension field $\F_{q^M}$ for $i=1,\dots,\mu$ and $j=1,\dots,n$. We write $\Gamma := (\alpha_{1,1},\alpha_{1,2},...,\alpha_{\mu,n})$.
The $[\mu n,\mu(n-r)-s]$ code $\code(\mu ,n,r,s,\H_0,\Gamma)$ is given by the $(r\mu +s)\times \mu n$ parity-check matrix 
  \begin{align*}
   \H =
     \begin{bsmallmatrix}
      \H_0 & \mathbf{0} &\hdots & \mathbf{0} \\
      \mathbf{0} & \H_0 &\hdots & \mathbf{0} \\
      \vdots & \vdots &\ddots & \cdots \\
      \mathbf{0}& \mathbf{0} &\hdots & \H_0 \\
      \H_1 & \H_2 &\hdots & \H_{\mu } \\
    \end{bsmallmatrix} 
  \end{align*}
  where for $1\leq j\leq \mu$,
  \begin{align*}
    \H_{j} =
    \begin{bsmallmatrix}
      \alpha_{j,1} & \alpha_{j,2} &  \hdots & \alpha_{j,n} \\
      \alpha_{j,1}^{q} & \alpha_{j,2}^{q} &  \hdots & \alpha_{j,n}^{q} \\
     \vdots & \vdots & \ddots & \vdots \\
     \alpha_{j,1}^{q^{s-1}} & \alpha_{j,2}^{q^{s-1}} &  \hdots & \alpha_{j,n}^{q^{s-1}} \\
    \end{bsmallmatrix} \ .
  \end{align*}
Whether this code is PMDS/SD depends on the choice of $\Gamma$.
In \cite{gabrys2018constructions} the authors present multiple methods of finding such a sequence and in this work we will focus on their main result given in \cite[Section~IV-A]{gabrys2018constructions}.

The original construction in \cite{gabrys2018constructions} used for $\H_0$ a parity-check matrix of a specific Reed--Solomon code.
It can be seen from the proof of \cite[Lemma~2]{gabrys2018constructions} that this restriction is not necessary in general and that an arbitrary MDS parity-check matrix gives a PMDS code as well.
This slight generalization is necessary for the construction in the next subsection.

\subsection{New Construction of $d$-MSR PMDS Array Codes}

In the following, we construct PMDS array codes, in which each row constitutes a, possibly different, PMDS code as in \cite{gabrys2018constructions} (cf.~Subsection~\ref{ssec:recap_PMDS_arbitrary_r_s}) and the local array codes are Ye--Barg codes \cite{ye2017optimalRepair}.

\begin{construction}\label{con:PMDSany}
Let $q,M,\mu,n,r,s,d$ be positive integers where
\begin{itemize}
    \item $r \leq n$
    \item $s \leq (n-r)\mu$
    \item $q$ a prime power
    \item $q \geq bn$, where $b=d+1-(n-r)$
    \item $\ell = b^n$ .
\end{itemize}
Let $\F_q$ be a finite field and let $\alpha_{i,j} \in \F_{q^M}$ for $i=1,\dots,\mu$ and $j=1,\dots,n$, and write $\Gamma := (\alpha_{1,1},\alpha_{1,2},\dots,\alpha_{\mu,n})$.
Furthermore, let $\beta_{i,j}$ for $i=1,\dots,n$ and $j=1,\dots,b$ be distinct elements of $\F_q$ and write $\mathcal{B} := (\beta_{1,1},\beta_{1,2},\dots,\beta_{n,b})$.

We define the following $[\mu n, \mu(n-r)-s; \ell]_{q^M}$ array code $\code(\mu ,n,r,s,d,\mathcal{B},\Gamma;\ell)_{q^M}$ as
\begin{align*}
\left\{
\C = \begin{bsmallmatrix}
\c^{(0)} \\
\c^{(1)}\\
\vdots 	   \\
\c^{(\ell-1)} 
\end{bsmallmatrix} \in \F_{q^M}^{\ell \times \mu n} \, : \, \H^{(a)} \c^{(a)} = \0 \, \forall \, a=0,\dots,\ell-1
\right\},
\end{align*}
where the matrix $\H^{(a)}$ is defined as
\begin{align*}
    \H^{(a)} :=
    \begin{bsmallmatrix}
      \H_0^{(a)} & \mathbf{0} &\hdots & \mathbf{0} \\
      \mathbf{0} & \H_0^{(a)} &\hdots & \mathbf{0} \\
      \vdots & \vdots &\ddots & \cdots \\
      \mathbf{0}& \mathbf{0} &\hdots & \H_0^{(a)} \\
      \H_1 & \H_2 & \hdots & \H_{\mu }
    \end{bsmallmatrix} \in \F_{q^M}^{r \mu +s \times \mu n} \ ,
  \end{align*}
  and for each $a \in [0,\ell-1]$ with $a = \sum_{i = 1}^{n} a_{i} b^{i-1}$, we have
      \begin{align}\label{eq:PMDSanyLocal}
      \H_0^{(a)}  =
      \begin{bsmallmatrix}
        1&1& \hdots & 1 \\
        \beta_{1,a_1} & \beta_{2,a_2} & \hdots & \beta_{n,a_{n}}\\
        \vphantom{\int\limits^x}\smash{\vdots} & \vphantom{\int\limits^x}\smash{\vdots} & & \vphantom{1}\smash{\vdots} \\
        \beta_{1,a_1}^{r-1} & \beta_{2,a_2}^{r-1} & \hdots & \beta_{n,a_{n}}^{r-1} 
      \end{bsmallmatrix} \in \F_{q}^{r \times n} \ ,
  \end{align}
Further, for $1\leq j \leq \mu$, define
\begin{align*}
\H_{j} :=
\begin{bsmallmatrix}
  \alpha_{j,1} & \alpha_{j,2} &  \hdots & \alpha_{j,n} \\
  \alpha_{j,1}^{q} & \alpha_{j,2}^{q} &  \hdots & \alpha_{j,n}^{q} \\
 \vdots & \vdots & \ddots & \vdots \\
 \alpha_{j,1}^{q^{s-1}} & \alpha_{j,2}^{q^{s-1}} &  \hdots & \alpha_{j,n}^{q^{s-1}} \\
\end{bsmallmatrix} \in \F_{q^M}^{s \times n}.
\end{align*}
\end{construction}
By the choice of the $\beta_{i,j}$, the local codes of the codes given by Construction~\ref{con:PMDSany} are Ye--Barg codes.
As in \cite{gabrys2018constructions}, we need to choose the vector $\Gamma$ in a suitable way to obtain PMDS array codes.

\begin{lemma}\label{lem:PMDS_any}
Let $\alpha_{1,1},\dots,\alpha_{\mu,n}$ be chosen such that any subset of $(r+1)s$ elements of the $\alpha_{i,j}$ is linearly independent over $\F_q$. Then, $\code(\mu ,n,r,s,d,\mathcal{B},\Gamma;\ell)_{q^M}$ from Construction~\ref{con:PMDSany} is a $d$-MSR PMDS array code.
\end{lemma}

\begin{proof}
The proof combines the ideas of \cite[Lemma~2]{gabrys2018constructions}, \cite[Corollary~5]{gabrys2018constructions}, and \cite[Lemma~7]{gabrys2018constructions}.
Let
\begin{align*}
\C = \begin{bsmallmatrix}
\c_1^{(0)} & \c_2^{(0)} & \dots  & \c_\mu^{(0)} \\
\c_1^{(1)} & \c_2^{(1)} & \dots  & \c_\mu^{(1)} \\
\vdots 	   & \vdots     & \ddots & \vdots \\
\c_1^{(\ell-1)} & \c_2^{(\ell-1)} & \dots  & \c_\mu^{(\ell-1)}
\end{bsmallmatrix} \in \F_{q^M}^{\ell \times \mu n}
\end{align*}
be a codeword of $\code(\mu ,n,r,s,\mathcal{L},\delta,\Gamma)$ (here, we group the rows of the codeword in blocks of length $n$, i.e., $\c^{(a)}_i \in \F_{q^M}^{n}$).
By definition, for all $i=1,\dots,\mu$ and $a=0,\dots,\ell-1$, we have
\begin{equation}
\H^{(a)} {\c_i^{(a)}}^\top = \0, \label{eq:any_parities_first_parity_eq}
\end{equation}
where $\H^{(a)}$ is the parity-check matrix of an $[n,n-r]$ MDS code.
Furthermore, with $\alphaVec_{i} := [\alpha_{i,1},\alpha_{i,2},\dots,\alpha_{i,n}]$, we have
\begin{align}
\sum_{i=1}^{\mu} \alphaVec_{i}^{q^j} {\c_i^{(a)}}^\top = 0, \label{eq:any_parities_second_parity_eq}
\end{align}
for all $j=0,\dots,s-1$ and $a=0,\dots,\ell-1$.

Let $S := [\s_1,\dots,\s_\mu]$ be of the form
\begin{align*}
\s_i = [s_{i,1},\dots,s_{i,r}] \in [n]^r, \quad s_{i,1} < s_{i,2} < \dots < s_{i,r}.
\end{align*}
Denote by $\bar{\s}_i$ the vector in $[n]^{n-r}$ that contains, again in increasing order, the entries of $[n]$ that are not contained in $\s_i$.
The positions $\s_i$ correspond to the puncturing patterns $E_i$ in the definition of PMDS array codes (cf.~Definition~\ref{def:pmds}). We need to show that for each such vector $S$, the array code punctured at these positions in each local group, gives an $[n-r\mu,n-r\mu-s;\ell]$ MDS array code.

For a vector $\x$ of length $n$, let $\x_{\s_i}$ and $\x_{\sbar_i}$ be the vectors of length $r$ and $n-r$ containing the entries of $\x$ indexed by the entries of $\s_i$ and $\sbar_i$, respectively.
Furthermore, for a vector $\y = [\y_1,\dots,\y_\mu]$ of length $n\mu$, let $\y^S$ denote the puncturing of $\y$ at all entries of $S$, i.e.,
\begin{equation*}
    \y^S = \big[ (\y_1)_{\sbar_1}, \dots, (\y_\mu)_{\sbar_\mu}\big].
\end{equation*}
Let $\H$ be a parity-check matrix of an MDS code of length $n$ and dimension $n-r$. Then, the columns of $\H$ indexed by $\s_i$, denoted by $\H_{\s_i}$, are invertible and we have for any codeword $\x$ of the code
\begin{equation*}
\0 = \H \x^\top = \H_{\s_i} \x_{\s_i}^\top + \H_{\bar{\s}_i} \x_{\bar{\s}_i}^\top \quad
\Rightarrow \quad \x_{\s_i}^\top = \H_{\s_i}^{-1} \H_{\bar{\s}_i} \x_{\bar{\s}_i}^\top
\end{equation*}

Hence, it directly follows from \eqref{eq:any_parities_first_parity_eq} that 
\begin{align*}
{(\c_i)_{\s_i}^{(a)}}^\top = {\H_{\s_i}^{(a)}}^{-1} \H_{\bar{\s}_i}^{(a)} (\c_i)_{\bar{\s}_i}^\top,
\end{align*}
and by \eqref{eq:any_parities_second_parity_eq} that (note that $\H^{(a)}$ has entries in $\F_q$, so $\H^{(a)} = {\H^{(a)}}^{q^j}$ for any $j$)
\begin{align*}
0 &= \sum_{i=1}^{\mu} \alphaVec_{i}^{q^j} {\c_i^{(a)}}^\top = \sum_{i=1}^{\mu} (\alphaVec_{i})_{\s_i}^{q^j} {(\c_i)_{\s_i}^{(a)}}^\top + (\alphaVec_{i})_{\sbar_i}^{q^j} {(\c_i)_{\sbar_i}^{(a)}}^\top \\
  &= \sum_{i=1}^{\mu} \Big[\underbrace{ (\alpha_i)_{\s_i} {\H_{\s_i}^{(a)}}^{-1} \H_{\bar{\s}_i}^{(a)} + (\alpha_i)_{\sbar_i}}_{=: \, \gammaVec_{\s_i}^{(a)}} \Big]^{q^{j}} {(\c_i)_{\sbar_i}^{(a)}}^\top.
\end{align*}
Thus, the vector
$\c^{S} = \big[(\c_1)_{\sbar_1}^{(a)}, (\c_2)_{\sbar_2}^{(a)}, \dots, (\c_\mu)_{\sbar_\mu}^{(a)} \big]$, which is the $a$-th row of a codeword punctured at the positions in $S$,
is contained in a code with parity-check matrix
\begin{align*}
\H_\gamma^{(a)} := 
\begin{bsmallmatrix}
{\gammaVec_S^{(a)}}^{q^0} \\
{\gammaVec_S^{(a)}}^{q^1} \\
\vdots \\
{\gammaVec_S^{(a)}}^{q^{s-1}} \\
\end{bsmallmatrix},
\end{align*}
where
\begin{align*}
{\gammaVec_S^{(a)}} := \Big[ \gammaVec_{\s_1}^{(a)}, \gammaVec_{\s_2}^{(a)}, \dots, \gammaVec_{\s_\mu}^{(a)} \Big] \in \F_{q^M}^{\mu(n-r)}.
\end{align*}

By definition, we have
\begin{align*}
\gammaVec_{\s_i}^{(a)} = (\alpha_i)_{\s_i} {\H_{\s_i}^{(a)}}^{-1} \H_{\bar{\s}_i}^{(a)} + (\alpha_i)_{\sbar_i}.
\end{align*}
Since ${\H_{\s_i}^{(a)}}^{-1} \H_{\bar{\s}_i}^{(a)}$ is an $r \times (n-r)$ matrix, each entry of $\gammaVec_{\s_i}^{(a)}$, and thus each entry of $\gammaVec_S^{(a)}$, is linear combination of at most $r+1$ of the $\alpha_{i,j}$. Furthermore, each such linear combination contains, non-trivially, one element from $\alpha_{i,j}$ (namely the corresponding entry in $(\alpha_i)_{\sbar_i}$) that appears only in this linear combination.
Any set of $s$ entries from $\gammaVec_S^{(a)}$ depend on at most $s(r+1)$ of the $\alpha_{i,j}$, which are linearly independent by the independence assumption.
Hence, the $s$ entries from $\gammaVec_S^{(a)}$ are also linearly independent over $\F_q$.
This means that any $s$ columns of the parity-check matrix $\H_\gamma^{(a)}$ are linearly independent and $\H_\gamma^{(a)}$ is a parity-check matrix of an $[n\mu-r\mu,n\mu-r\mu-s]_{q^M}$ MDS code.
Thus, the overall code is a PMDS array code.

It remains to show that the local codes are $d$-MSR codes. The proof is equivalent to the proof of the locally $d$-MSR property of Theorem~\ref{thm:PMDSMSR}. First, note that the $\beta_{i,j}$ in Construction~\ref{con:PMDSany} are the (distinct) elements $\beta^0,\beta^1,...,\beta^{rn-1}$. Now consider the $j$-th local array code. The $a$-th row fulfills the parity check equations given in (\ref{eq:PMDSanyLocal}) and since all elements $\beta_{i,j}$ are distinct, it is immediate that the local group is an $[n,n-r; b^{n}]$ Ye-Barg $d$-MSR code as in Definition~\ref{def:yeBarg}.
\end{proof}

Similar to Construction~\ref{con:SDcodesLocalRegeneration} in the previous section, we give the following upper bound on the minimal field size for which $d$-MSR PMDS codes of the form as in Construction~\ref{con:PMDSany} exist.

\begin{theorem}
Let $\mu,n,r,s$ be fixed.
There is a $d$-MSR PMDS array code as in Construction~\ref{con:PMDSany} of field size 
\begin{equation*}
    q^M \leq 2n(d+1-(n-r)) (2n\mu)^{s(r+1)-1}
\end{equation*}
 and subpacketization
\begin{equation*}
    \ell = \big[d+1-(n-r)\big]^n.
\end{equation*}
\end{theorem}

\begin{proof}
We choose $q$ and $M$ large enough such that we can ensure that suitable sequences $\alpha_{i,j}$ and $\beta_{i,j}$ exist.
A sufficient condition for the existence of the $\beta_{i,j}$ is $q \geq n(d+1-(n-r))$.
Thus, we can choose $q$ to be the smallest prime power greater or equal to $n(d+1-(n-r))$, which is at most $q \leq 2n(d+1-(n-r))$ by Bertrand’s postulate.

For the $\alpha_{i,j}$, it is a bit more involved.
By Lemma~\ref{lem:PMDS_any}, it suffices to find $n\mu$ elements from $\F_{q^M}$ of which any subset of $s(r+1)$ elements is linearly independent.
We use the same idea as in \cite[Lemma~7]{gabrys2018constructions}: take the columns of a parity-check matrix of a $\Code[n\mu,n\mu-M,s(r+1)+1]_q$ code and interpret each column $\F_q^M$ as an element of $\F_{q^M}$.
The independence condition is then fulfilled due to the choice of minimum distance.

The remaining question is for which $M$ and $q$ a code with parameters $[n\mu,n\mu-M,s(r+1)+1]_q$ exists.
We use the result in \cite[Problem~8.9]{roth2006introduction}, which we can reformulate in our terms as: For any $n' = q^a-1$, there exists a code with parameters $[n',n'-M,s(r+1)+1]_q$, where
\begin{align*}
    M \leq 1+ \big[s(r+1)-1\big]a.
\end{align*}
Choose $a$ to be the smallest integer with $n' = q^a-1 \geq n\mu$.
Note that there is such an $a$ with $q^a-1 \leq 2n\mu-1$, i.e., $a \leq \log_q(2n\mu)$.
Hence, there is an $[n',n'-M,s(r+1)+1]_q$ code with $M \leq 1+ \big[s(r+1)-1\big]\log_q(2n\mu)$.
Shortening the codes gives an $[n\mu,n\mu-M,s(r+1)+1]_q$ code with $M \leq 1+ \big[s(r+1)-1\big]\log_q(2n\mu)$.
\end{proof}


\IEEEtriggeratref{3}



\bibliographystyle{IEEEtran}
\bibliography{main.bib}
\end{document}